\documentclass[a4paper]{article}
\usepackage[utf8]{inputenc}

\usepackage{graphicx,latexsym,amsmath,xcolor,wrapfig}
\usepackage[rflt]{floatflt}

\usepackage[hidelinks]{hyperref}

\newcommand{\etal}{{et al.}}

\newcommand{\spath}[1]{\langle #1 \rangle} 

\graphicspath{{figures/}}

\usepackage{amsthm}

\newtheorem{lemma}{Lemma}

\newtheorem{corollary}{Corollary}
\newtheorem{theorem}{Theorem}

\newcommand{\floor}[1]{\ensuremath{\left \lfloor #1 \right \rfloor}}
\newcommand{\ceil}[1]{\ensuremath{\left \lceil #1 \right \rceil}}

\begin{document}
\title{Packing Plane Spanning Trees and Paths in Complete Geometric Graphs%
\thanks{This work was presented at the 26th Canadian Conference on Computational Geometry (CCCG 2014), Halifax, Nova Scotia, Canada, 2014.
The journal version appeared in Information Processing Letters, 124 (2017), 35--41, \url{https://doi.org/10.1016/j.ipl.2017.04.006}.}
}

\author{Oswin Aichholzer%
\thanks{Institute of Software Technology, Graz University of Technology, Austria.
\texttt{oaich@ist.tugraz.at}}
\and Thomas Hackl%
\thanks{Institute of Software Technology, Graz University of Technology, Austria.
\texttt{thackl@ist.tugraz.at}}
\and Matias Korman\thanks{
Tohoku University, Sendai, Japan.
\texttt{mati@dais.is.tohoku.ac.jp}}
\and Marc van Kreveld\thanks{
Department of Information and Computing Sciences, Utrecht University, the Netherlands.
\texttt{m.j.vankreveld@uu.nl}}
\and Maarten L\"offler\thanks{
Department of Information and Computing Sciences, Utrecht University, the Netherlands.
\texttt{m.loffler@uu.nl}}
\and Alexander Pilz\thanks{
Department of Computer Science, ETH Zurich, Switzerland.
\texttt{alexander.pilz@inf.ethz.ch}}
\and Bettina Speckmann\thanks{
Department for Mathematics and Computer Science, TU Eindhoven, the Netherlands.
\texttt{b.speckmann@tue.nl}}
\and Emo Welzl\thanks{
Department of Computer Science, ETH Zurich, Switzerland.
\texttt{welzl@inf.ethz.ch}}
}
\maketitle

\begin{abstract}
We consider the following question: How many edge-disjoint plane spanning trees are contained in a complete geometric graph $GK_n$ on any set $S$ of $n$ points in general position in the plane?
We show that this number is in $\Omega(\sqrt{n})$.
Further, we consider variants of this problem by bounding the diameter and the degree of the trees (in particular considering spanning paths).
\end{abstract}


\section{Introduction}
\label{sec:intro}

A \emph{geometric graph} $G = (S, E)$ consists of a set of vertices $S$, which are points in general position in the plane, and a set of edges $E$ which are straight-line connections between two of these points. A long-standing open question is the following: Does every complete geometric graph with $2n$ vertices have a partition of its edges into $n$ plane spanning trees?
For \emph{complete convex geometric graphs} (where all vertices lie in convex position), a positive answer to this question follows from a result by Bernhart and Kainen~\cite{book_embeddings} (see~\cite{partitions_into_trees}). Bose et al.~\cite{partitions_into_trees} gave a characterization of the solutions; for complete convex geometric graphs all spanning trees can, but do not have to, be spanning paths.
They also described a sufficient condition generalizing the convex case and considered a relaxation where the trees are not required to be spanning.

We consider a closely related question: How many edge-disjoint plane spanning trees are contained in a \emph{complete geometric graph} $GK_n$ on any set $S$ of $n$ points in general position in the plane?
 In Section~\ref{sec:trees} we show how to combine a construction by Bose et al.~\cite{partitions_into_trees} with a result by Aronov \etal~\cite{crossing_families} to prove that $GK_n$ contains $\Omega(\sqrt{n})$ edge-disjoint plane spanning trees. Furthermore, if the convex hull of $S$ contains $h$ vertices then we can argue that $GK_n$ contains at least $\left\lfloor\frac{h}{2}\right\rfloor$ edge-disjoint plane spanning trees. We also show that $GK_n$ contains at least 2 plane edge-disjoint spanning trees if $n \geq 4$ and at least 3 edge-disjoint spanning trees if $n \geq 6$.

In Section~\ref{sec:paths} we study the special case of spanning paths. In particular, we first consider the ``regular wheel configuration'', that is, a set of points
$W_{2n}$ which consists of $2n-1$ points regularly spaced on a circle $C$ and a point at the center of $C$. Let $GW_{2n}$ be the complete geometric graph on $W_{2n}$. We can argue that $GW_{2n}$ can be partitioned into $n$ spanning trees. But surprisingly, if $n \geq 3$ then none of these trees can be paths.
If the ``hub'' of the wheel is moved close to the convex hull, then all $n$ spanning trees can be paths.
This raises the following interesting open question:
When does this transition happen and is it gradual?
That is, does the number of spanning paths increase whenever the hub passes over certain diagonals? Note, though, that spanning paths can of course be used in packings which are not partitions.
More specifically, $GW_{2n}$ always contains $n-1$ spanning paths. Only when we ask for a complete partition of the edges we cannot use even a single spanning path.

On the positive side we argue that $GK_n$ contains at least 2 edge-disjoint spanning paths if $n \geq 4$. Obviously it would be desirable to extend our argument to $3$ or more paths or to develop a different line of reasoning to prove that $GK_n$ always contains many paths. Alternatively, it would be very interesting to find point sets which contain only few edge-disjoint plane spanning paths.

We also study packings of edge-disjoint planar spanning trees that have bounded vertex degree and bounded diameter. In particular, in Section~\ref{sec:degree} we show that for any $k \leq \sqrt{n/12}$ any set of $n$ points has $k$ edge-disjoint plane spanning trees with maximum vertex degree $O(k^2)$ and diameter $O(\log(n/k^2))$.

\paragraph{Related work}
A classic related problem in extremal graph theory is the following.
For general geometric graphs, what is the maximum number $f(k,n)$ such that there exists a geometric graph~$G$ of~$n$ vertices and $f(k,n)$ edges such that $G$ contains no $k$ disjoint edges?
Erd\H{o}s~\cite{disjoint_pair} showed that for all $n \geq 3$, $f(2,n) = n$, i.e., any geometric graph with $n+1$ edges contains a disjoint pair.
For general $k$, T\'oth and Valtr~\cite{toth_valtr} gave the lower and upper bounds of $3/2(k-1)n - 2k^2 \leq f(k+1, n) \leq k^3(n+1)$, and also showed that $4n-9 \leq f(4,n) \leq 8.5n$.
\v{C}ern\'y~\cite{cerny} proved $f(3,n) \leq \floor{2.5n}$.
More specifically, the existence of certain plane subgraphs has been investigated.
K\'arolyi, Pach, and T\'oth~\cite{monochromatic_pst} showed that any edge 2-coloring of a complete geometric graph $GK_n$ admits a monochromatic plane spanning tree.
\v{C}ern\'y et al.~\cite{noncrossing_hamiltonian} also considered the existence of plane spanning trees in geometric graphs.
They showed that after removing any set of at most $(1/2 \sqrt{2})\sqrt{n}$ edges from any $GK_n$, the resulting graph still contains a plane spanning path.
Aichholzer et al.~\cite{noncrossing_configurations} considered perfect matchings, subtrees and triangulations as plane subgraphs;
further references to similar results can be found in~\cite{noncrossing_configurations}.
For any geometric graph~$G$, Rivera-Campo~\cite{five_points} showed that if any subgraph of $G$ induced by five vertices has a plane spanning tree, then $G$ as well has a plane spanning tree.
Keller et al.~\cite{keller} gave a characterization of the smallest subgraphs of any~$GK_n$ that share at least one edge with any plane spanning tree of~$GK_n$ (so-called \emph{blockers}).
They showed that if a subgraph~$G$ is a blocker for all plane spanning trees of diameter at most four, then $G$ blocks all plane spanning subgraphs;
if the vertices of $GK_n$ are in convex position, the result already holds for a diameter of at most three.

Also the number of plane spanning trees attracted interest, analogously to classic results on the number of spanning trees (the \emph{tree density}) in general graphs.
Nash-Williams~\cite{nash_williams} and Tutte~\cite{tutte} independently showed that a graph~$G$ has a tree density of $k$ if $|E_P(G)| \geq k (|P|-1)$ for every partition $P$ of $V(G)$, where $E_P(G)$ denotes the set of edges between different members of~$P$.
This was used by Kundu~\cite{kundu} to relate the tree density in general graphs to their edge-connectivity: any $k$-edge-connected graph has at least $\ceil{k-1/2}$ edge-disjoint spanning trees.

Our problem is also closely related to the concept of \emph{$k$-book embeddings} of topological graphs, where, informally, the vertices are considered to be on the spine of a book and each edge of the graph is either on the spine or on exactly one of the $k$ pages, such that no two edges cross.
The \emph{book thickness} of a graph~$G$ is the smallest number~$k$ for which there exists a $k$-book embedding of~$G$.
Bernhart and Kainen~\cite[Theorem~3.4]{book_embeddings} showed that, for $n \geq 4$ vertices, the book thickness of the complete graph is $\ceil{n/2}$.
Their construction of $\floor{n/2}$ edge-disjoint paths directly carries over to packing the same amount of plane spanning paths in the complete convex geometric graph~\cite{partitions_into_trees}.

A concept between graph-theoretical thickness and book thickness was later developed by Dillencourt, Eppstein, and Hirschberg~\cite{thickness}:
given an abstract graph~$G$, the \emph{geometric thickness} of~$G$ is the smallest number~$k$ such that there exists a straight-line drawing of the graph that can be partitioned into $k$ plane subgraphs.
They showed that the geometric thickness of the (abstract) complete graph is between $\ceil{(n/5.646)+0.342}$ and $\ceil{n/4}$.

Since the initial presentation of this work, the problem has attracted further attention.
Most prominently, the lower bound on the number of plane edge-disjoint spanning trees has been improved to~$\lfloor n/3 \rfloor$ by Garc{\'\i}a~\cite{g-nepstg-15}.
Schnider~\cite{double_stars} considers the special case of double stars (i.e., trees with only two interior nodes), showing that a partition into such trees does not always exist, and provides necessary as well as sufficient conditions for its existence.


\section{Packing Spanning Trees}\label{sec:trees}

Recall that $GK_{n}$ is the complete geometric graph on any set $S$ of $n$ points in general position in the plane.

\begin{theorem}\label{thm_sqrtlayers}
$GK_{n}$ contains $\Omega(\sqrt{n})$ edge-disjoint plane spanning trees.
\end{theorem}
\begin{proof}
Let $S$ be a set of $n$ points in the plane, and let $F$ be a set of $k$ edges (pairs of points of $S$) such that each pair of edges in $F$ has an interior crossing.
The set $F$ is called a \emph{crossing family}.
We claim that there exists a set of $k$ edge-disjoint plane spanning trees on $S$. We use a construction similar to the \emph {double stars} by Bose \etal{}~\cite{partitions_into_trees}. For each edge $e = \overline{pq} \in F$, let $\ell_e$ be the supporting line of $e$. We connect all points to the left of $\ell_e$ to $p$, and all points to the right of $\ell_e$ to $q$. These edges together with $e$ form a tree~$T_e$ (see Figure~\ref{fig:doublestar}).

To see that this yields $k$ edge-disjoint trees, consider two trees $T_{\overline{pq}}$ and $T_{\overline{rs}}$. Suppose some edge is in both trees. Then one of its endpoints must be $p$ or $q$, and the other endpoint must be $r$ or $s$. However, if $r$ lies to the left of $\ell_{\overline{pq}}$, then $\overline {pr}$ and $\overline {qs}$ are in $T_{\overline{pq}}$ and $\overline {ps}$ and $\overline {qr}$ are in $T_{\overline{rs}}$, and vice versa if $r$ lies to the right of $\ell_{\overline{pq}}$.

Aronov \etal~\cite{crossing_families} showed that any set of $n$ points contains a crossing family of size $\sqrt{n/12}$.
The theorem follows immediately.
\end{proof}

\begin{figure}[h]
\centering
\includegraphics{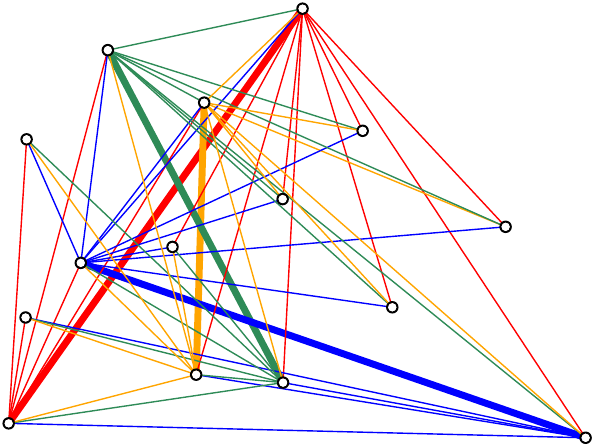}
\caption{A set of $15$ points with $4$ pairwise crossing edges.}
\label{fig:doublestar}
\end{figure}

\smallskip\noindent
In a set of $h$ points in convex position, there is always a crossing family of size $\floor{h/2}$.
The proof of Theorem~\ref{thm_sqrtlayers} therefore immediately implies the following.

\begin{corollary}
The complete graph of a set $S$ of $n$ points, of which $h$ are in convex position, contains at least $\floor{h/2}$ edge-disjoint plane spanning trees.
\end{corollary}

\begin{theorem}\label{thm:2trees}
If $n \geq 4$ then $GK_{n}$ contains at least 2 edge-disjoint plane spanning trees.
\end{theorem}
\begin{proof}
\begin{figure}[htb]
\centering
\includegraphics{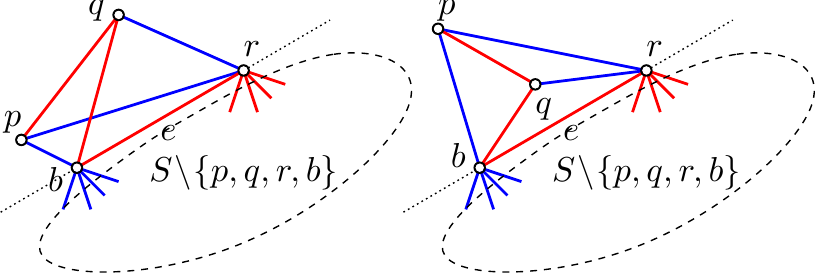}
\caption{The two cases for constructing two edge-disjoint plane
  spanning trees on $S$.}
\label{fig:2trees}
\end{figure}
Let $S$ be a set of $n$ points in the plane and let $e=rb$ be an edge spanned by $S$ having exactly 2 points ($p$ and $q$) of $S$ on one side (i.e., on one side of the straight line supporting $e$).
The set $\{p,q,r,b\}$ is either in convex position (Case~1; see \figurename~\ref{fig:2trees}~(left)) or forms a triangle with one interior point (Case~2; see \figurename~\ref{fig:2trees}~(right)).
Note that $e$ has to be an edge of the convex hull of $\{p,q,r,b\}$.
W.l.o.g., let $pqrb$ be the convex polygon in Case~1 and let $q$ be the point inside the triangle $prb$ in Case~2.
In both cases we construct two edge-disjoint spanning trees on $\{p,q,r,b\}$, $\spath{q,r,p,b}$ (blue) and $\spath{p,q,b,r}$ (red).
To get two edge-disjoint spanning trees on $S$ we connect all points of $S\setminus\{p,q,r,b\}$ with $b$ (for the blue tree) and with $r$ (for the red tree).
\end{proof}

Note that the proof of Theorem~\ref{thm_sqrtlayers} also immediately implies Theorem~\ref{thm:2trees} for $n\geq5$, because then there always exists a pair of crossing edges in $GK_{n}$.
For $n=4$ the two cases for $\{p,q,r,b\}$ shown in \figurename~\ref{fig:2trees} serve as a proof.

\begin{lemma}\label{lem:6points}
$GK_{6}$ contains 3 edge-disjoint plane spanning trees.
\end{lemma}
\begin{proof}
For $n=6$ there exist 16 combinatorially different point sets (order types)~\cite{aak-eotsp-01a}.
It is easy to check that each of these 16 cases allows for 3 edge-disjoint plane spanning trees packed on $GK_{6}$ (see \figurename~\ref{fig:6points}).
\end{proof}
%
%
Using the order type database for small point sets~\cite{a-otdb-06}
it can be easily checked that $GK_{8}$ and $GK_{9}$ each contain 4 edge-disjoint plane spanning trees,
and that $GK_{10}$ contains 5 edge-disjoint plane spanning trees.
(The latter has been obtained by reducing the set of order types to so-called crossing-maximal ones, as characterized in~\cite{pw-oo-15}.)

\begin{figure}[h]
\centering
\includegraphics{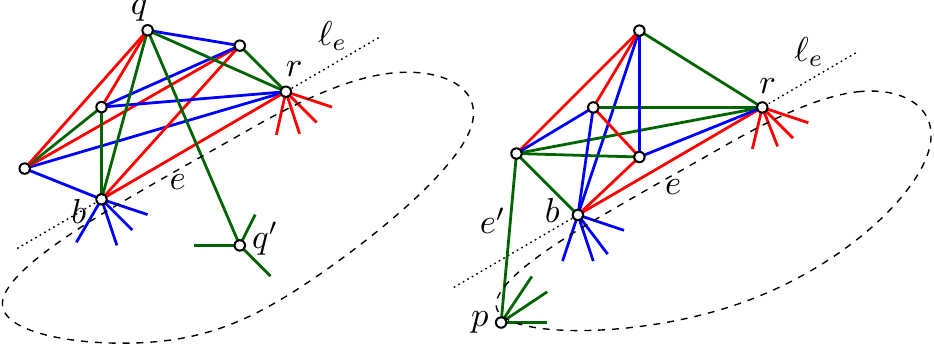}
\caption{Two examples depicting the construction of three edge-disjoint plane
  spanning trees on $S$.}
\label{fig:3trees}
\end{figure}

\begin{theorem}
If $n \geq 6$ then $GK_{n}$ contains at least 3 edge-disjoint plane spanning trees.
\end{theorem}
\begin{proof}
Let $S$ be a set of $n$ points in the plane and let $e=rb$ be an edge spanned by $S$ having exactly 4 points of $S$ on one side (i.e., on one side of the straight line $\ell_e$ supporting $e$).
Let $S'$ be the set of 6 points containing $r$, $b$, and the exactly 4 points on one side of $e$.
By Lemma~\ref{lem:6points}, $S'$ contains 3 edge-disjoint plane spanning trees.
For simplicity we call them red, blue, and green.
W.l.o.g., assume that $e$ is part of the red tree.
Note that each point of $S'$ is incident to all three trees, and that $r$ and $b$ are extremal points for $S\setminus (S'\setminus\{r,b\})$.
We construct a red and a blue plane spanning tree by connecting $r$ and $b$, respectively, with all points in $S\setminus S'$.

Next we construct the third (green) plane spanning tree on $S$.
Note that the green plane spanning tree on $S'$ can be completed to a triangulation $T$.
Let $q$ be the point of $S'\setminus\{r,b\}$ such that $qrb$ is a triangle in $T$.
Observe that any edge incident to $q$ and crossing $e$ does not cross a green edge.

Assume that there exists a point $q'\in(S\setminus S')$ such that the edge $qq'$ crosses~$e$.
Then we connect $q$ and $q'$ with a green edge and complete the green
plane spanning tree by connecting all points in
$S\setminus(S'\cup\{q'\})$ with $q'$.
See \figurename~\ref{fig:3trees}~(left).

If such a point $q'$ does not exist, then there has to exist an edge $e'$ of the convex hull of $S$, such that $e'$ crosses $\ell_e$.
Denote by $p$ the endpoint of $e'$ in $S\setminus S'$.
We color $e'$ green and complete the green
plane spanning tree by connecting all points in
$S\setminus(S'\cup\{p\})$ with $p$.
See \figurename~\ref{fig:3trees}~(right).
\end{proof}

\section{Packing Spanning Paths}\label{sec:paths}

Let $W_{2n}$ be a set of $2n$ points in the ``regular wheel configuration'' in the plane. $W_{2n}$ consists of $2n-1$ points regularly spaced on a circle $C$ and a point at the center of $C$. Let $GW_{2n}$ be the complete geometric graph on $W_{2n}$.

\begin{figure*}[t]
\centering
\includegraphics[width=\columnwidth]{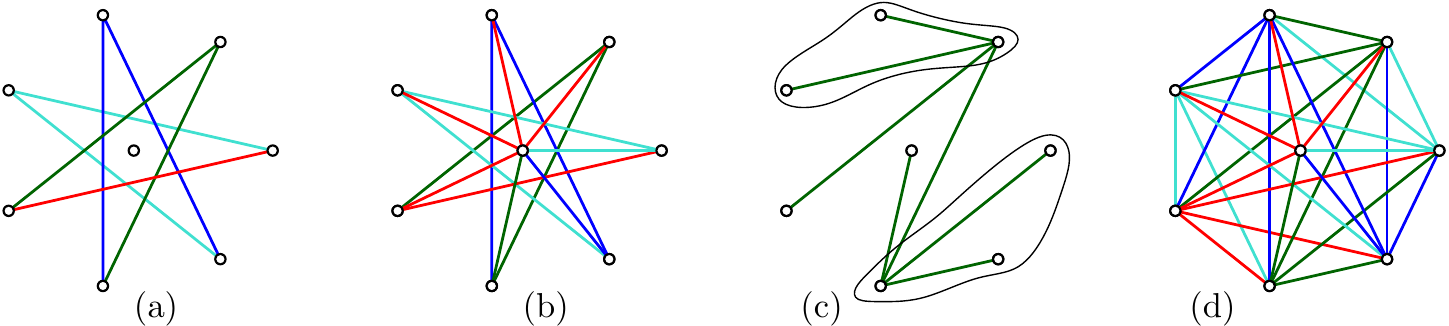}
\caption{The graph $GW_{2n}$ cannot have plane spanning paths if it is partitioned into
plane spanning trees.}
\label{fig:wheel}
\end{figure*}

\begin{theorem}
$GW_{2n}$ can be partitioned into $n$ spanning trees. If $n\geq 3$ then none of these trees can be a path.
\end{theorem}
\begin{proof}
In the following, we color the edges of $GW_{2n}$ that each class is plane and spanning. %
Let $v_{0}$ be the central vertex and let the other vertices be $v_1,\ldots, v_{2n-1}$ in cyclic order.
The complete graph has edges of varying length between the vertices $v_1,\ldots, v_{2n-1}$, and we can
use $E_1,\ldots,E_{n-1}$ to denote the length classes of the edges, from short to long. The edges involving $v_{0}$ are
called the radial edges.
There are $2n-1$ edges in each length class and also $2n-1$ radial edges.

We first consider the length class $E_{n-1}$, then the radial edges, and then $E_{n-2},\ldots,E_1$, and see how we must
color these edges to produce plane spanning trees.

Given that there are $2n-1$ edges in $E_{n-1}$, to be divided over $n$ colors, and every non-adjacent
pair of edges intersect, we will get these edges in $n-1$ pairs and one singleton, see Figure~\ref{fig:wheel}(a).
Call the color of the singleton edge in $E_{n-1}$ red.
The pairs must be two adjacent edges (they have a shared vertex), forming a wedge with point $v_{0}$ in between and at
least one point to each side of the wedge if there are at least six points. This immediately shows that all spanning trees with non-red color are not paths. To show that a red spanning tree also cannot be a path, we
observe that $v_0$ can have at most one edge in each non-red color (otherwise we make a cycle or an intersection
within that color). Therefore, it must have $n$ incident red edges,
showing that the red spanning tree is not a path either if $n\geq 3$ (Figure~\ref{fig:wheel}(b)).

We proceed to show that the geometric graph contains $n$ plane spanning trees. We color the radial edges by using
the red color $n$ times. There are two options when we do not have crossings or cycles, and they are symmetric.
The remaining radial edges get the other $n-1$ colors, one for each, and such that a path of length $3$ appears
in each color. Then we assign the edges in $E_{n-2},\ldots,E_1$ a color at once.
We make $2n-1$ fans, one for each of $v_1,\ldots, v_{2n-1}$, consisting of one edge of each length
class (there are two choices: clockwise and counterclockwise), see Figure~\ref{fig:wheel}(c) for the two fans of one color.
Each fan can be assigned a color so that all spanning  trees are isomorphic balanced double stars, completing the
partitioning into $n$ plane spanning trees (Figure~\ref{fig:wheel}(d)).
\end{proof}
\begin{figure}[b]
\centering
\includegraphics{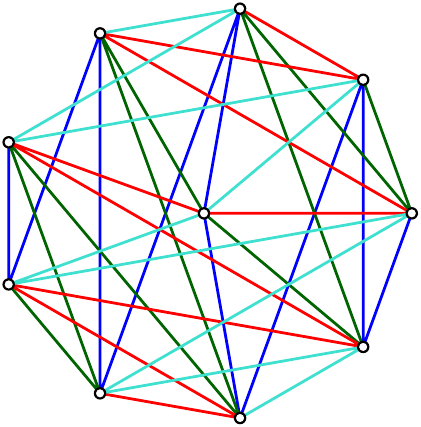}
\caption{$GW_{2n}$ contains $n-1$ plane spanning paths.}
\label{fig:pathsinwheel}
\end{figure}
Interestingly, $GW_{2n}$ contains $n-1$ plane spanning paths, via the zigzag construction
used for points in convex position (as described in~\cite{partitions_into_trees}).
When the path passes the center point, it picks it up using two radial edges
instead of a long edge, see Figure~\ref{fig:pathsinwheel}. But to get one more plane spanning tree in $GW_{2n}$,
all paths must be trees.

We now return to $GK_{n}$, the complete geometric graph on any set $S$ of $n$ points in general position in the plane.

\begin{theorem}
If $n \geq 4$ then $GK_{n}$ contains at least 2 edge-disjoint plane spanning paths.
\end{theorem}
\begin{proof}
Let $S$ be a set of $n$ points in the plane and let $p$ be an extremal point of $S$.
Order the points of $S\setminus\{p\}$ clockwise around $p$.
Partition $S\setminus\{p\}$ into two (disjoint) sets $A$ and $B$, such that $A\cup B = S\setminus\{p\}$ and $|B|-1 \leq |A| \leq |B|$.
We denote by $\ell$ a line through $p$ (but no other point of $S$) that is separating $A$ from $B$ (see \figurename~\ref{fig:partition}).

\begin{figure}[htb]
\centering
\includegraphics{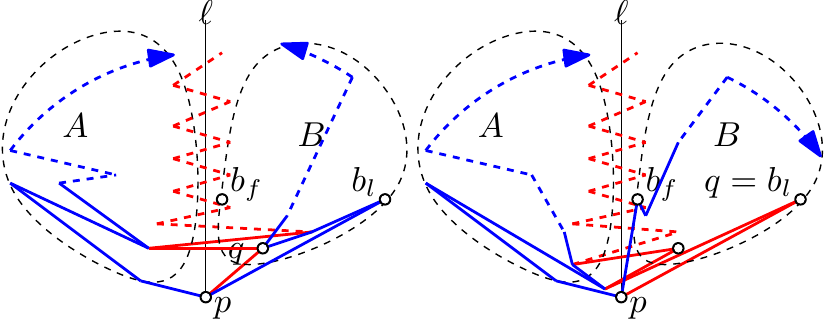}
\caption{Partition of $S\setminus\{p\}$ and the two edge-disjoint
 spanning paths. Left: $q\neq b_l$. Right: $q=b_l$.}
\label{fig:partition}
\end{figure}

We will construct the two edge-disjoint paths, for simplicity call them red and blue.
The red path (${\cal R}=G(V,E_1)$) we simply construct as a plane zigzag path starting at $p$, with a point $q$ in $B$ as a second point, and with every edge of $\cal R$, except $pq$, intersecting $l$.
(An algorithm for constructing such a zigzag path is described by Hershberger and Suri~\cite{hershberger_suri}, see also Abellanas et al.~\cite{bipartite_embeddings}.)

The blue path (${\cal B}=G(V,E_2)$) consists of two subpaths, ${\cal B}_A$ and ${\cal B}_B$, joined at $p$.
Observe that no red edge (edge of $\cal R$) connects two points of $A\cup\{p\}$ or two points of $B$.
Thus, any (blue) path completely contained in $A\cup\{p\}$ is edge-disjoint to $\cal R$.
We choose the path starting at $p$ and connecting the points of $A$ in clockwise order around $p$ for ${\cal B}_A$.

Let $b_f$ and $b_l$ be the first and last, respectively, point of $B$ in clockwise order around $p$.
If $q=b_l$ then we connect $p$ with $b_f$ and continue on the points of $B\setminus\{b_f\}$ in clockwise order around $p$ for ${\cal B}_B$ (see \figurename~\ref{fig:partition}~(right)).
Otherwise, we construct ${\cal B}_B$ with $pb_l$ as the first edge and then finish the path by connecting the points of $B\setminus\{b_l\}$ in counter clockwise order around $p$ (see \figurename~\ref{fig:partition}~(left)).

Connecting ${\cal B}_A$ and ${\cal B}_B$ at $p$ results in the plane spanning path ${\cal B}$ that is edge-disjoint to the plane spanning path ${\cal R}$.
\end{proof}

\section{Packing Spanning Trees with low Degree}\label{sec:degree}

The edge-disjoint plane spanning trees we studied in the previous sections are somehow extreme in terms of vertex degree. The trees constructed in Section~\ref{sec:trees} always contain at least one vertex of degree $\Omega(n)$, while in Section~\ref{sec:paths} we consider spanning paths. Thus the question arises if intermediate results are possible.
In the following, we obtain a trade-off between the number of edge-disjoint spanning trees and the maximum degree of each vertex.

\begin{figure}[b]
\centering
\includegraphics{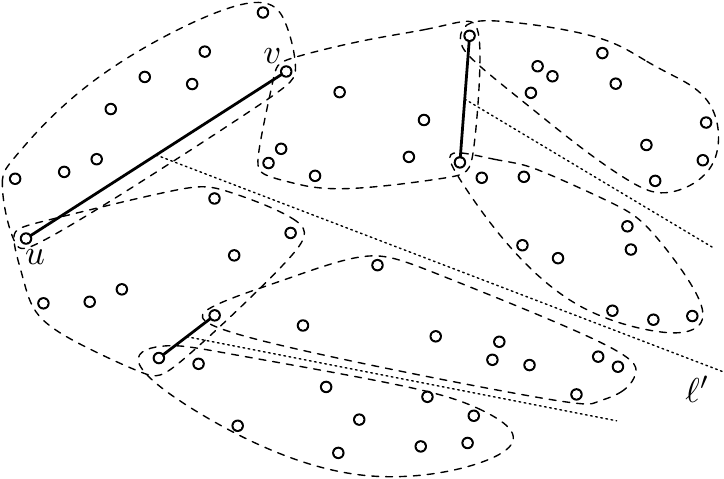}
\caption{The hierarchical clustering strategy.}
\label{fig_decomp}
\end{figure}

\begin{figure*}[t]
\centering
\includegraphics[width=\columnwidth]{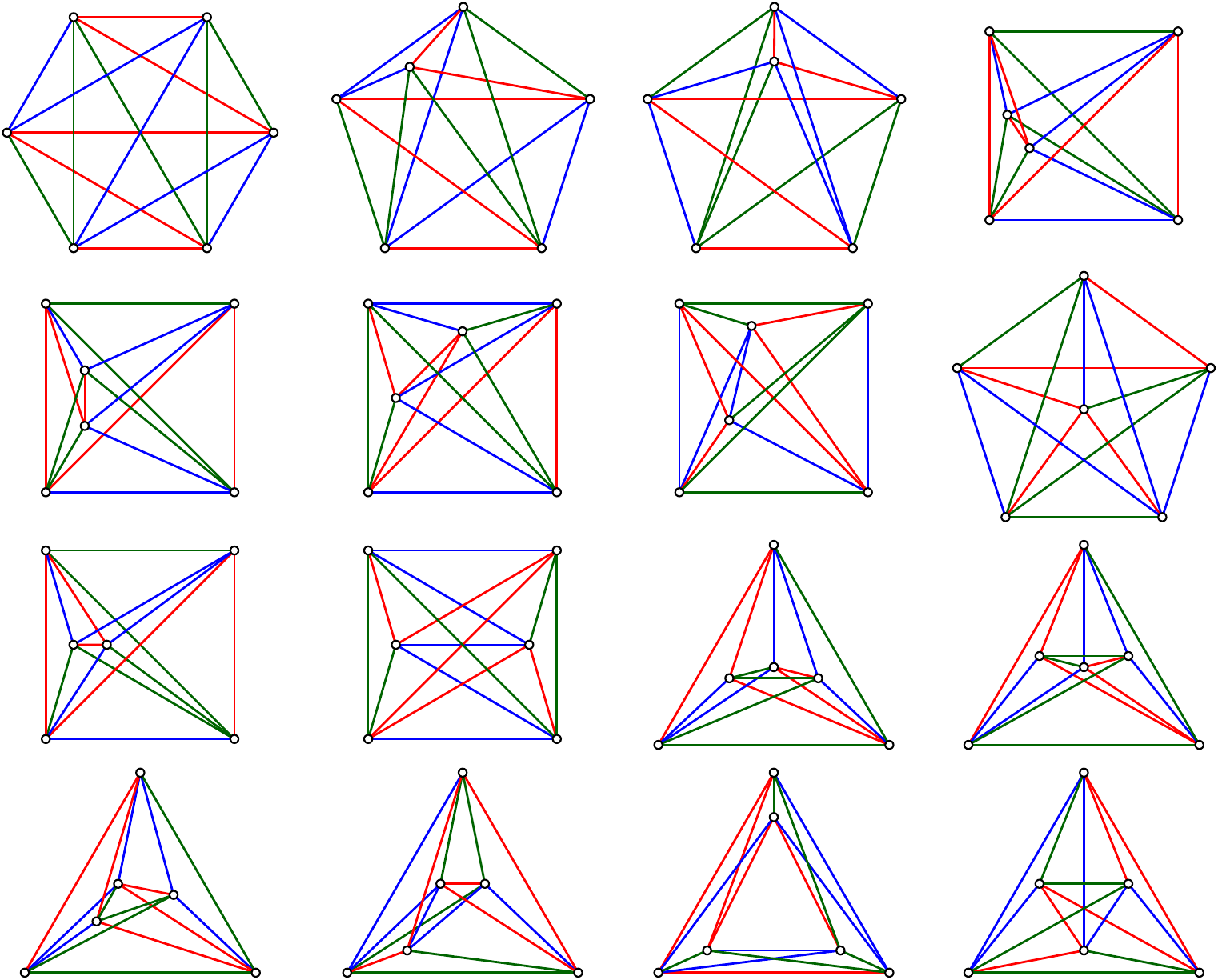}
\caption{The 16 combinatorially different point sets
for $n=6$~\cite{aak-eotsp-01a,a-otdb-06}, with 3 edge-disjoint plane spanning trees each.}
\label{fig:6points}
\end{figure*}

\begin{theorem}\label{theo_hierarch}
For any set $S$ of $n$ points and $k\leq \sqrt{n/12}$ there exist $k$ edge-disjoint plane spanning trees $T_1,\ldots, T_k$ on~$S$ such that the maximum degree of any tree is in $O(k^2)$. %
Also, the diameter of each tree is in $O(\log(n/k^2))$.
\end{theorem}
\begin{proof}
The general idea of the proof is to ``peel off'' small clusters of points and connect each of the clusters with $k$ edge-disjoint spanning trees independently.
Consider a $(12k^2-2)$-edge, i.e., an edge $uv$, $u,v \in S$, such that exactly $12k^2-2$ points of~$S$ are strictly to the left of the directed line~$\ell$ through~$uv$.
Consider the set $C_1$ of these $12k^2$ points and construct $k$ edge-disjoint plane spanning trees of~$C_1$ using Theorem~\ref{thm_sqrtlayers}.
Now consider the midpoint between $u$ and $v$.
Let $\ell'$ be a line through that midpoint that splits the remaining point set $S \setminus C_1$ into two subsets $S_u$ and $S_v$, each containing at most $\ceil{(n-12k^2)/2}$ points.

Since the two subsets are separated by $\ell'$, we can recursively repeat a similar process in the two subsets independently.
That is, pick a $(12k^2-2)$-edge $u'v'$ of $S_u \cup \{u\}$ such that $u$ is contained among the $12k^2$ points separated by $u'v'$ but is not an endpoint of the edge (such an edge must always exist).
We construct $k$ plane spanning trees on this subset, which are connected to the spanning trees of~$C_1$ via~$u$. We treat $S_v \cup \{v\}$ analogously (see Figure~\ref{fig_decomp}).
The recursion stops when we are not able to partition the remaining points into two sets of size at least $12k^2-1$;
here, we simply add the remaining points of the subset to the last cluster.
Note that this cluster must have between $12k^2$ and $36k^2-3$ points, thus we can still create $k$ edge-disjoint spanning trees using Theorem~\ref{thm_sqrtlayers}.

We construct the $k$ spanning trees of $S$ by assigning one of the spanning trees of each cluster arbitrarily to each of the trees $T_1, \ldots, T_k$.
We claim that the resulting trees are indeed spanning:
By construction, each tree is spanning in the cluster;
hence points of the same cluster will be connected in $T_i$ (for all $i\leq k$).
Moreover, the hierarchical construction certifies that each cluster shares a point with the cluster constructed in the previous step of induction.
Likewise, planarity of each tree is guaranteed.

We obtain at most $N = \floor{n/(12k^2-1)}$ clusters which are arranged such that they form a balanced binary tree with $C_1$ as root.
Note that the spanning trees constructed in the proof of Theorem~\ref{thm_sqrtlayers} have diameter~$3$.
Thus, the diameter of each spanning tree is at most $6\ceil{\log_2 N}$.
The degree bound follows from the fact that any point of $S$ can only belong to at most two clusters (and each cluster has $\Theta(k^2)$ points).
\end{proof}

\paragraph{Acknowledgments.} Research was initiated during the 10th European Research Week on Geometric Graphs (GGWeek 2013), Illgau, Switzerland. We would like to thank all participants of the 10th European Research Week on Geometric Graphs for fruitful discussions.
T.~H.\ was supported by the Austrian Science Fund (FWF): P23629-N18.
M.~K.\ was partially supported by MEXT KAKENHI No.~17K12635.
A.~P.\ is supported by an Erwin Schr\"odinger fellowship, Austrian Science Fund (FWF): J-3847-N35.


\small
\bibliographystyle{abbrv}
\bibliography{packing}

\begin{thebibliography}{10}

\bibitem{bipartite_embeddings}
M.~Abellanas, J.~Garcia-Lopez, G.~Hern{\'a}ndez-Pe{\~n}alver, M.~Noy, and P.~A.
  Ramos.
\newblock Bipartite embeddings of trees in the plane.
\newblock {\em Discrete Applied Mathematics}, 93(2--3):141--148, 1999.

\bibitem{a-otdb-06}
O.~Aichholzer.
\newblock Order types.
\newblock
  \url{http://www.ist.tugraz.at/aichholzer/research/rp/triangulations/ordertypes/}.
\newblock Last retrieved on May 7, 2014.

\bibitem{aak-eotsp-01a}
O.~Aichholzer, F.~Aurenhammer, and H.~Krasser.
\newblock {{Enumerating Order Types for Small Point Sets with Applications}}.
\newblock {\em Order}, 19:265--281, 2002.

\bibitem{noncrossing_configurations}
O.~Aichholzer, S.~Cabello, R.~Fabila~Monroy, D.~Flores-Pe{\~n}aloza, T.~Hackl,
  C.~Huemer, F.~Hurtado, and D.~R. Wood.
\newblock Edge-removal and non-crossing configurations in geometric graphs.
\newblock {\em Discrete Mathematics {\&} Theoretical Computer Science},
  12(1):75--86, 2010.

\bibitem{crossing_families}
B.~Aronov, P.~Erd{\H{o}}s, W.~Goddard, D.~J. Kleitman, M.~Klugerman, J.~Pach,
  and L.~J. Schulman.
\newblock Crossing families.
\newblock {\em Combinatorica}, 14(2):127--134, 1994.

\bibitem{book_embeddings}
F.~Bernhart and P.~C. Kainen.
\newblock The book thickness of a graph.
\newblock {\em Journal of Combinatorial Theory, Series B}, 27(3):320--331,
  1979.

\bibitem{partitions_into_trees}
P.~Bose, F.~Hurtado, E.~Rivera-Campo, and D.~R. Wood.
\newblock Partitions of complete geometric graphs into plane trees.
\newblock {\em Comput. Geom.}, 34(2):116--125, 2006.

\bibitem{cerny}
J.~{\v{C}ern\'y}.
\newblock Geometric graphs with no three disjoint edges.
\newblock {\em Discrete {\&} Computational Geometry}, 34(4):679--695, 2005.

\bibitem{noncrossing_hamiltonian}
J.~{\v{C}ern{\'y}}, Z.~Dvorak, V.~Jel\'{\i}nek, and J.~K{\'a}ra.
\newblock Noncrossing {H}amiltonian paths in geometric graphs.
\newblock {\em Discrete Applied Mathematics}, 155(9):1096--1105, 2007.

\bibitem{thickness}
M.~B. Dillencourt, D.~Eppstein, and D.~S. Hirschberg.
\newblock Geometric thickness of complete graphs.
\newblock {\em J. Graph Algorithms Appl.}, 4(3):5--17, 2000.

\bibitem{disjoint_pair}
P.~Erd{\H{o}}s.
\newblock On sets of distances of $n$ points.
\newblock {\em The American Mathematical Monthly}, 77(7):738--740, 1970.

\bibitem{g-nepstg-15}
A.~Garc{\'\i}a.
\newblock On the number of edge-disjoint plane spanning trees of {$K_n$}.
\newblock Unpublished manuscript, 2015.

\bibitem{hershberger_suri}
J.~Hershberger and S.~Suri.
\newblock Applications of a semi-dynamic convex hull algorithm.
\newblock {\em BIT}, 32(2):249--267, 1992.

\bibitem{monochromatic_pst}
G.~K{\'a}rolyi, J.~Pach, and G.~T{\'o}th.
\newblock Ramsey-type results for geometric graphs, {I}.
\newblock {\em Discrete {\&} Computational Geometry}, 18(3):247--255, 1997.

\bibitem{keller}
C.~{Keller}, M.~A. {Perles}, E.~{Rivera-Campo}, and V.~{Urrutia-Galicia}.
\newblock Blockers for noncrossing spanning trees in complete geometric graphs.
\newblock In J.~Pach, editor, {\em Thirty Essays on Geometric Graph Theory},
  volume~29 of {\em Algorithms and Combinatorics}, pages 383--398. Springer,
  2012.

\bibitem{kundu}
S.~Kundu.
\newblock Bounds on the number of disjoint spanning trees.
\newblock {\em Journal of Combinatorial Theory, Series B}, 17(2):199--203,
  1974.

\bibitem{nash_williams}
C.~{\relax St}. J.~A. Nash-Williams.
\newblock Edge-disjoint spanning trees of finite graphs.
\newblock {\em J. London Math. Soc.}, 36:445--450, 1961.

\bibitem{pw-oo-15}
A.~Pilz and E.~Welzl.
\newblock Order on order types.
\newblock In {\em Proc. 31st International Symposium on Computational Geometry
  (SoCG 2015)}, pages 285--299, 2015.

\bibitem{five_points}
E.~Rivera-Campo.
\newblock A note on the existence of plane spanning trees of geometric graphs.
\newblock In J.~Akiyama, M.~Kano, and M.~Urabe, editors, {\em JCDCG}, volume
  1763 of {\em LNCS}, pages 274--277. Springer, 1998.

\bibitem{double_stars}
P.~Schnider.
\newblock Packing plane spanning double stars into complete geometric graphs.
\newblock In {\em Proc.\ 32nd European Workshop on Computational Geometry
  (EuroCG '16)}, pages 91--94, 2016.

\bibitem{toth_valtr}
G.~T{\'o}th and P.~Valtr.
\newblock Geometric graphs with few disjoint edges.
\newblock {\em Discrete {\&} Computational Geometry}, 22(4):633--642, 1999.

\bibitem{tutte}
W.~T. Tutte.
\newblock On the problem of decomposing a graph into $n$ connected factors.
\newblock {\em J. London Math. Soc.}, 36:221--230, 1961.

\end{thebibliography}

\end{document}